\documentclass[journal,10pt]{IEEEtran}

\usepackage{graphicx}
\usepackage[T1]{fontenc}
\usepackage[bookmarks=false]{hyperref}

\usepackage{enumitem}
\usepackage[cmex10]{amsmath}
\usepackage{amssymb}
\usepackage{cite}
\usepackage{amsthm}
\usepackage{textcomp}
\usepackage{siunitx}
\usepackage{color}
\usepackage{subfigure}
\usepackage{cases}
\usepackage[font={small}]{caption}
\usepackage{bbm}
\usepackage{tcolorbox}

\newtheorem{theorem}{Theorem}
\newtheorem{lemma}{Lemma}

\newtheorem{corollary}{Corollary}

\theoremstyle{definition}

\newtheorem{definition}{Definition}

\theoremstyle{remark}
\newtheorem{remark}{Remark}

\setlist[description]{style=multiline}

\IEEEoverridecommandlockouts

\begin{document}
\sloppy

\title{A Scalable Framework for\\ Wireless Distributed Computing} 

\author{Songze~Li,~\IEEEmembership{Student~Member,~IEEE,}
        Qian~Yu,
        Mohammad~Ali~Maddah-Ali,~\IEEEmembership{Member,~IEEE,}
        and~A.~Salman~Avestimehr,~\IEEEmembership{Senior Member,~IEEE}
\thanks{S.~Li, Q.~Yu and A.S.~Avestimehr are with the Department of Electrical Engineering, University of Southern California, Los Angeles, CA, 90089, USA (e-mail: songzeli@usc.edu;  qyu880@usc.edu; avestimehr@ee.usc.edu).}
\thanks{M. A. Maddah-Ali is with Department of Electrical Engineering, Sharif University of Technology (e-mail: maddah\_ali@sharif.edu).}
\thanks{A part of this paper was presented in IEEE/ACM SEC, 2016~\cite{LYMA-SEC2016}. A shorter version of this paper was presented in IEEE GLOBECOM, 2016~\cite{LQMA_globecom}.} }

\maketitle

\begin{abstract}
We consider a wireless distributed computing system, in which multiple mobile users, connected wirelessly through an access point, collaborate to perform a computation task. In particular, users communicate with each other via the access point to exchange their locally computed intermediate computation results, which is known as \emph{data shuffling}. We propose a scalable framework for this system, in which the required communication bandwidth for data shuffling does not increase with the number of users in the network. The key idea is to utilize a particular repetitive pattern of placing the dataset (thus a particular repetitive pattern of intermediate computations), in order to provide coding opportunities at both the users and the access point, which reduce the required uplink communication bandwidth from users to access point and the downlink communication bandwidth from access point to users by factors that grow linearly with the number of users. We also demonstrate that the proposed dataset placement and coded shuffling schemes are optimal (i.e., achieve the minimum required shuffling load) for both a centralized setting and a decentralized setting, by developing tight information-theoretic lower bounds.
\end{abstract}


\section{Introduction}\label{sec:intro}
Recent years have witnessed a rapid growth of computationally intensive applications on mobile devices, such as mapping services, voice/image recognition, and augmented reality. The current trend for developing these applications is to offload computationally heavy tasks to a ``cloud'', which has greater computational resources. While this trend has its merits, there is also a critical need for enabling \emph{wireless distributed computing}, in which computation is carried out using the computational resources of a cluster of wireless devices collaboratively. Wireless distributed computing eliminates, or at least de-emphasizes, the need for a core computing environment (i.e., the cloud), which is critical in several important applications, such as autonomous control and navigation for vehicles and drones, in which access to the cloud can be very limited. Also as a special case of the emerging ``Fog computing architecture''~\cite{bonomi2012fog}, it is expected to provide significant advantages to users by improving the response latency, increasing their computing capabilities, and enabling complex applications in machine learning, data analytics, and autonomous operation (see e.g.,~\cite{drolia2013case,datla2012wireless,huerta2010virtual}).

Without the help from a centralized cloud, the local computing capability of a wireless device is often limited by its local storage size. For example, for a mobile navigation application in which a smart car wants to compute the fastest route to its destination, over a huge dataset containing the map information and the traffic conditions over a period of time, the local storage size of an individual car is too small to store the entire dataset, and hence individual processing is not feasible. However, using a wireless distributed computing framework, a group of smart cars, each storing a part of the dataset, can collaborate to meet their respective computational needs over the entire dataset.

The major challenge in developing a scalable framework for wireless distributed computing is the significant communication load, required to exchange the intermediate computation results among the mobile users. In fact, even when the processing nodes are connected via high-bandwidth inter-server communication bus links (e.g., a Facebook's Hadoop cluster), it is observed in~\cite{chowdhury2011managing} that 33\% of the job execution time is spent on data shuffling. The communication bottleneck is expected to get much more severe as we move to a wireless medium where the communication resources are much more scarce. More generally, as the network size increases, while the computation resources grow linearly with network size, the overall communication bandwidth is fixed and can become the bottleneck. This raises the following fundamental question.
\begin{tcolorbox}
Is there a \emph{scalable} framework for wireless distributed computing, in which the required communication load is fixed and independent of the number of users?
\end{tcolorbox}

Our main contribution is to provide an affirmative answer to this question by developing a framework for wireless distributed computing that utilizes redundant computations at the users, in order to create coding opportunities that reduce the required communication, achieving a scalable design. The developed framework can be considered as an extension of our previously proposed \emph{coded distributed computing} framework for a wireline setting in~\cite{LMA_all,LMA_ISIT16,li2016fundamental,li2017coded}, into the wireless distributed computing domain. To develop such a framework, we exploit three opportunities in conjunction:

\begin{enumerate}
\item \emph{Side-Information:} When a sub-task has been processed in more than one node, the resulting intermediate outcomes will be available in all those nodes as side-information. This provides some opportunities for \emph{coding} across the results and creates packets that are useful for multiple nodes. 
\item \emph{Coding:} We use coding to develop packets useful to more than one mobile users. This allows us to exploit the \emph{multicasting} environment of the wireless medium and save communication overhead. 
\item \emph{Multicasting:} Wireless medium by nature is a multicasting environment. It means that when a signal is transmitted, it can be heard by all the nodes. We exploit this phenomenon by creating and transmitting signals that help several user nodes simultaneously.
\end{enumerate}

\begin{figure}[htbp]
   \centering
   \subfigure[Uplink.]{\includegraphics[width=0.4\textwidth]{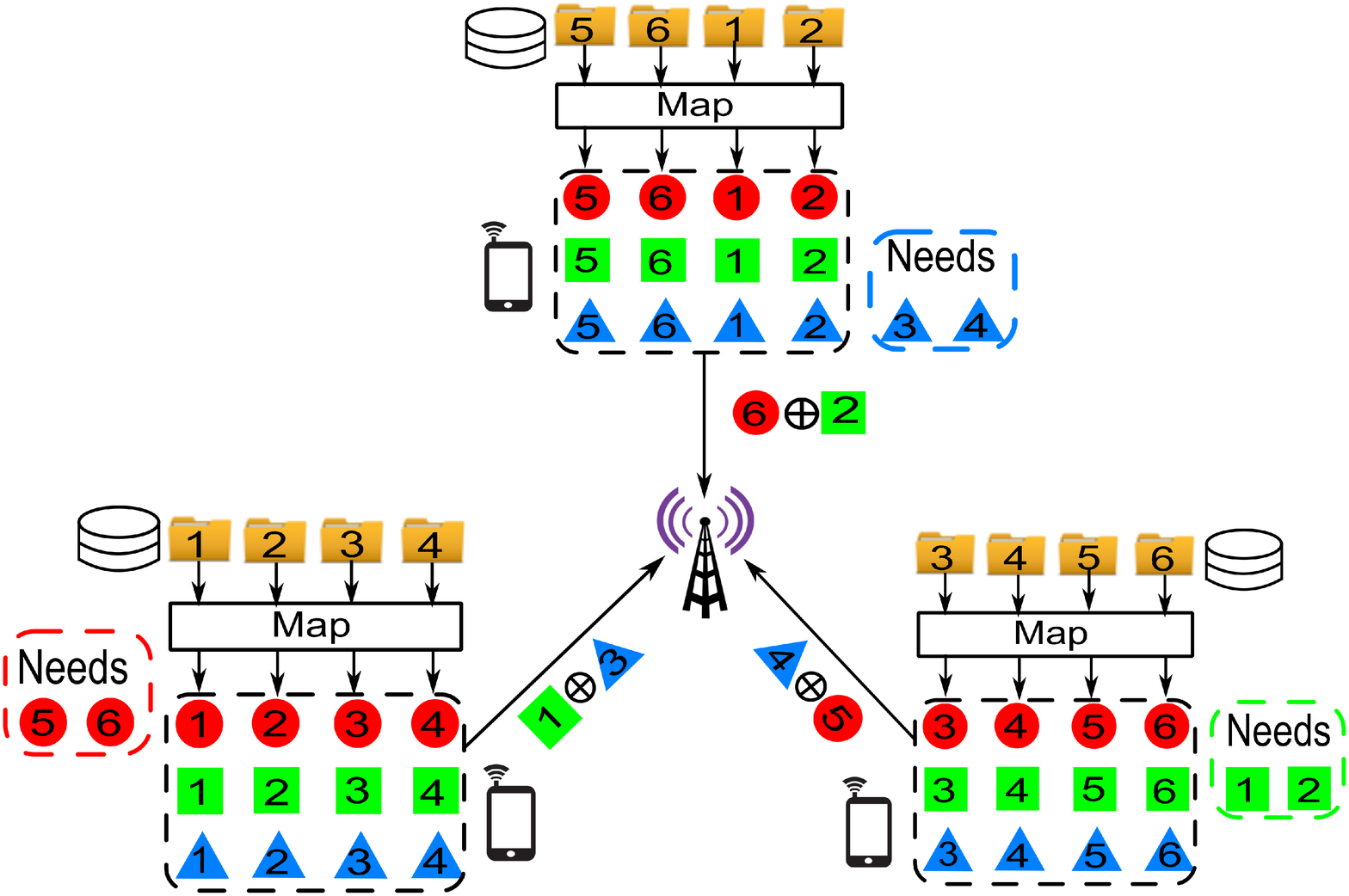}
       \label{fig:example_up}} 
    \hfill
   \subfigure[Downlink.]{\includegraphics[width=0.4\textwidth]{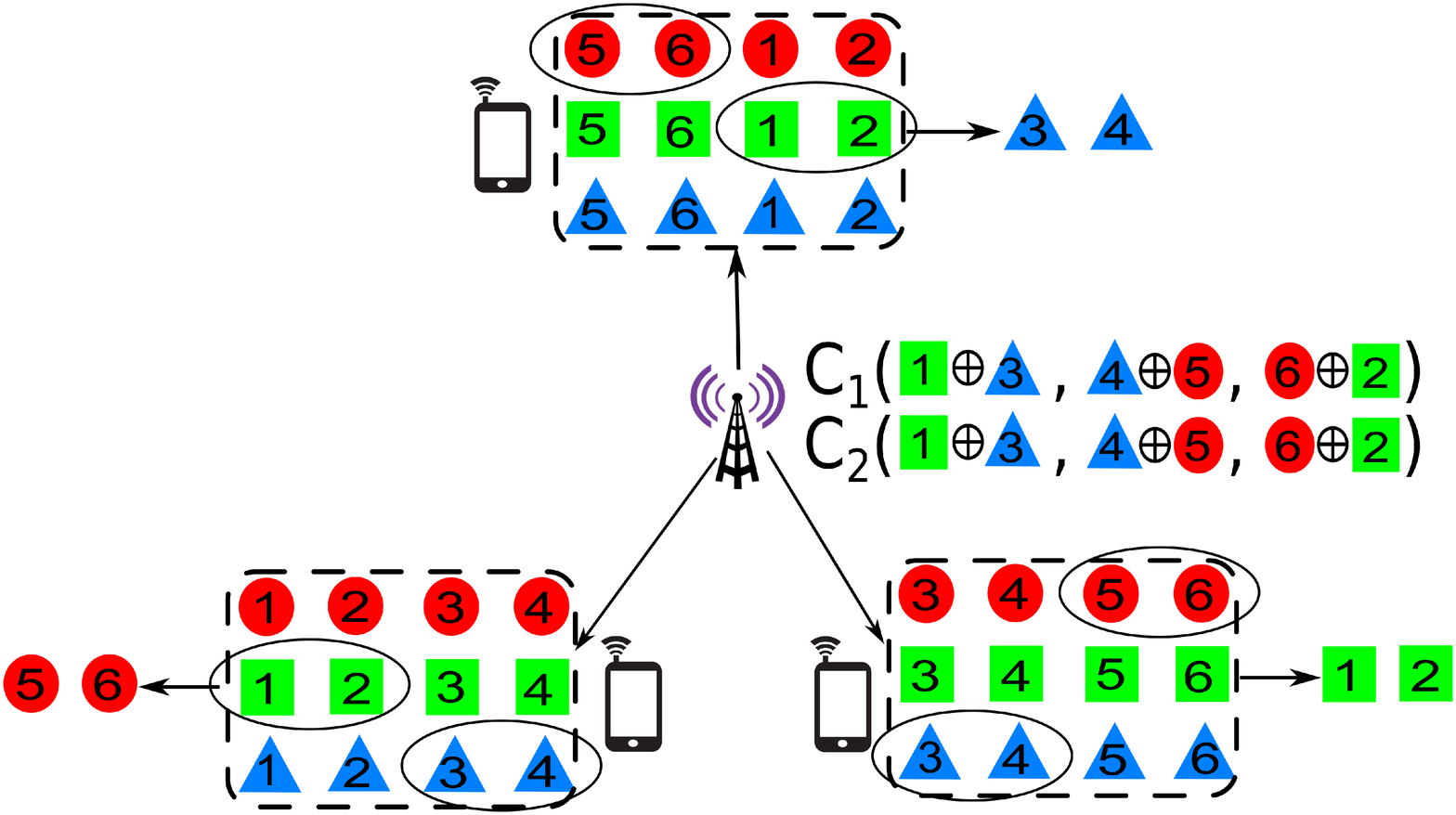}
       \label{fig:example_down}}
   \caption{Illustration of the CWDC scheme for an example of 3 mobile users.}
   \label{fig:example}
\end{figure}

\subsection{Motivating Example}\label{sec:example}
Let's first illustrate our scalable design of wireless distributed computing through an example. Consider a scenario where 3 mobile users want to run an application (e.g., image recognition). Each user has an input (e.g., an image) to process using a dataset (e.g., a feature repository of objects) provided by the application. However, the local memory of an individual user is too small to store the entire dataset, and they have to collaboratively perform the computation. The entire dataset consists of 6 equally-sized files, and each user can store at most 4 of them. The computation is performed distributedly following a commonly used MapReduce-like distributed computing structure (see e.g., MapReduce~\cite{dean2004mapreduce} and Spark~\cite{zaharia2010spark}). More specifically, every user computes a Map function, for each of the 3 inputs and each of the 4 files stored locally, generating 12 intermediate values. Then the users communicate with each other via an access point they all wirelessly connect to, which we call \emph{data shuffling}. After the data shuffling, each user knows the intermediate values of his own input in all 6 files, and passes them to a Reduce function to calculate the final output result.

During data shuffling, since each user already has 4 out of 6 intended intermediate values locally, she would need the remaining 2 from the other users. Thus, one would expect a communication load of 6 (in number of intermediate values) on the uplink from users to the access point and 6 on the downlink in which the access point simply forwards the intermediate values to the intended users. However, we can take advantage of the opportunities mentioned above to significantly reduce the communication loads. As illustrated in Fig.~\ref{fig:example}, through careful placement of the dataset into users' memories, we can design a coded communication scheme in which every user sends a bit-wise XOR, denoted by $\oplus$, of 2 intermediate values on the uplink, and then the access point, without decoding any individual value, simply generates 2 random linear combinations $C_1(\cdot,\cdot,\cdot)$ and $C_2(\cdot,\cdot,\cdot)$ of the received messages and broadcasts them to the users, simultaneously satisfying all data requests. Using this coded approach, we achieve an uplink communication load of 3 and a downlink communication load of 2.

We generalize the above example by designing a coded wireless distributed computing (CWDC) framework that applies to arbitrary type of applications, network size and storage size. In particular, we propose a specific dataset placement strategy, and a joint uplink-downlink communication scheme exploiting coding at both the mobile users and the access point. 

\begin{tcolorbox}
For a distributed computing application with $K$ users each can store $\mu$ fractions of the dataset, the proposed CWDC scheme achieves the (normalized) communication loads
\begin{align}
L_{\text{uplink}} \approx L_{\text{downlink}} \approx \tfrac{1}{\mu}-1.
\end{align}
\end{tcolorbox}

We note that the proposed scheme is scalable since the achieved communication loads are independent of $K$. As we show in Fig.~\ref{fig:load-storage}, compared with a conventional uncoded shuffling scheme with a communication load $\mu K \cdot (\tfrac{1}{\mu}-1)$ that explodes as the network expands, the proposed CWDC scheme reduces the load by a multiplicative factor of $\Theta(K)$.

\begin{figure}[htbp]
   \centering
   \subfigure[Uplink.]{\includegraphics[width=0.36\textwidth]{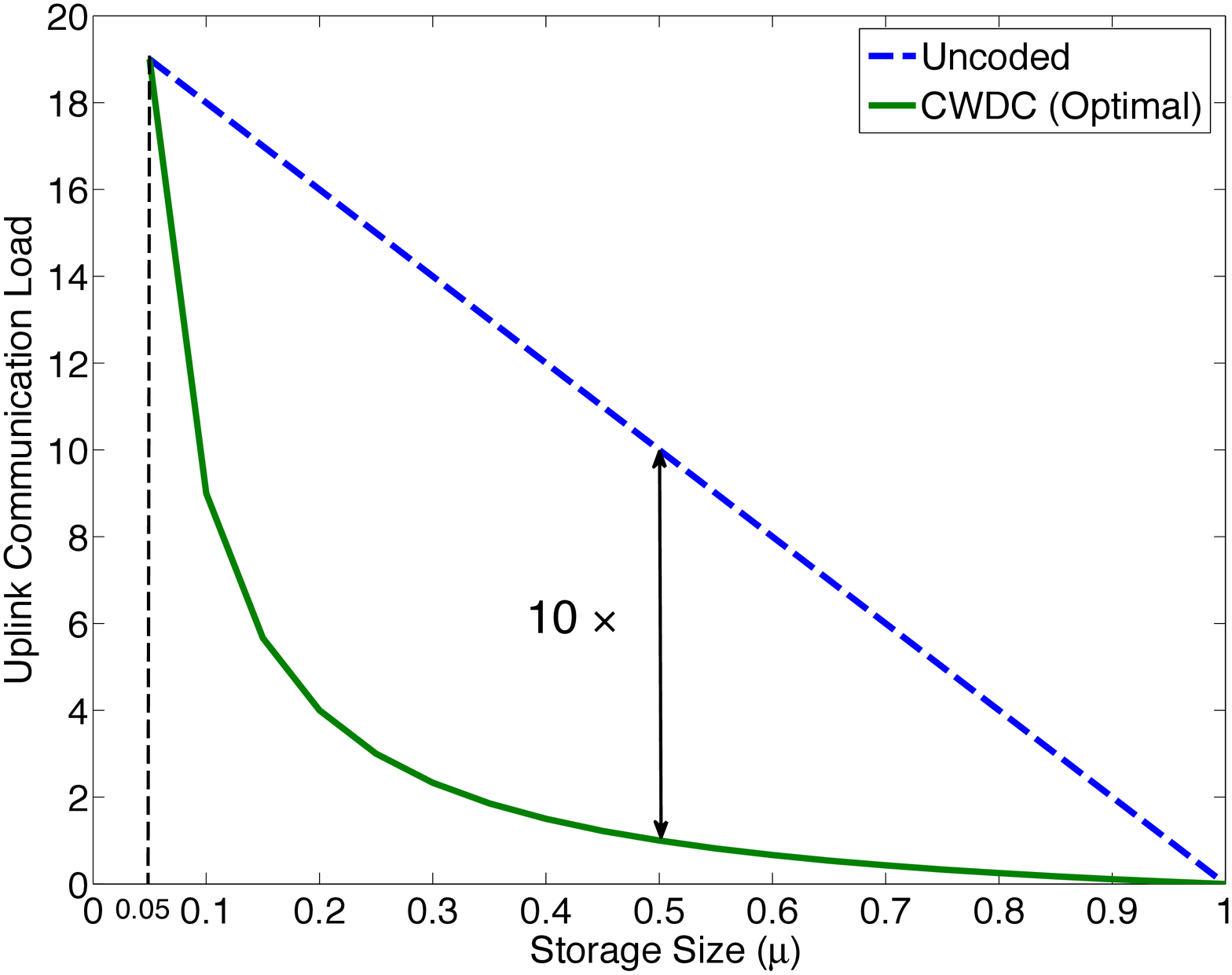}
       \label{fig:load_up}} 
   \subfigure[Downlink.]{\includegraphics[width=0.36\textwidth]{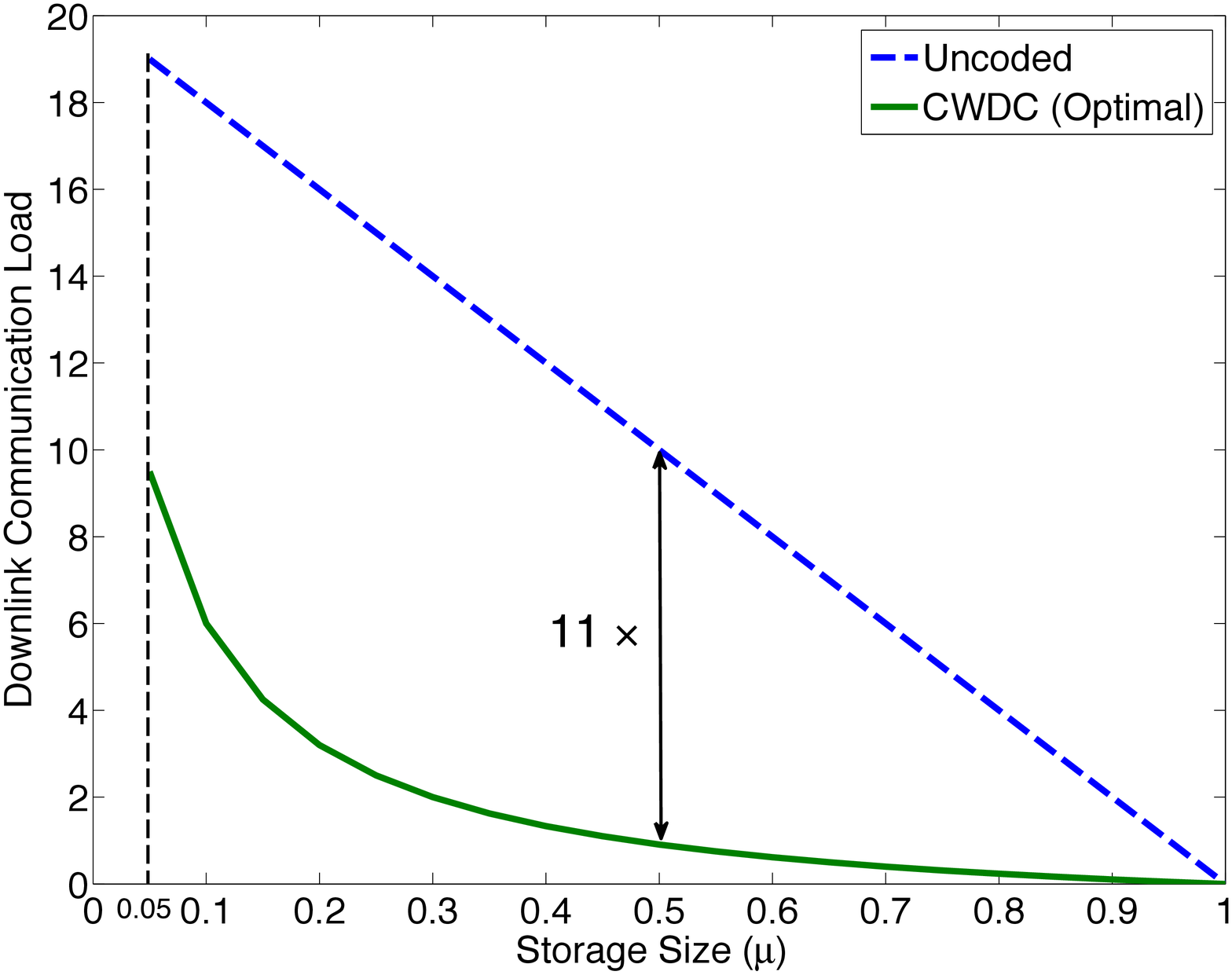}
       \label{fig:load_down}}
   \caption{Comparison of the communication loads achieved by the uncoded scheme with those achieved by the proposed CWDC scheme, for a network of $K=20$ users.
   }
   \label{fig:load-storage}
\end{figure}

We also extend our scalable design to a decentralized setting, in which the dataset placement is done at each user independently without knowing other collaborating users. For such a common scenario in mobile applications, we propose a decentralized scheme with a communication load close to that achieved in the centralized setting, particularly when the number of participating users is large.

Finally, we demonstrate that for both the centralized setting and the decentralized setting with a large number of users, the proposed CWDC schemes achieve the minimum possible communication loads that cannot be improved by any other scheme, by developing tight lower bounds.

\subsection{Prior Works}
The idea of applying coding in shuffling the intermediate results of MapReduce-like distributed computing frameworks was recently proposed in~\cite{LMA_all,LMA_ISIT16,li2016fundamental,li2017coded}, for a wireline scenario where the computing nodes can directly communicate with each other through a shared link. In this paper, we consider a wireless distributed computing environment, in which wireless computing nodes communicate via an access point. More specifically, we extend the coded distributed computing framework in~\cite{LMA_all,LMA_ISIT16,li2016fundamental,li2017coded} in the following aspects.
\begin{itemize}
\item We extend the wireline setting in~\cite{LMA_ISIT16,li2016fundamental} to a wireless setting, and develop the first scalable framework with constant communication loads for wireless distributed computing. 
\item During data shuffling, other than designing codes at the users for uplink communication, we also design novel optimum code at the access point for downlink communication.
\item We consider a decentralized setting that is of vital importance for wireless distributed computing, where each user has to decide its local storage content independently. We develop optimal decentralized dataset placement strategy and uplink-downlink communication scheme to achieve the minimum communication loads asymptotically.
\end{itemize} 

The idea of efficiently creating and exploiting coded multicasting opportunities was initially proposed in the context of cache networks in~\cite{maddah2014fundamental, maddah2013decentralized},  and extended in~\cite{ji2014fundamental, karamchandani2014hierarchical}, where caches pre-fetch part of the content in a way to enable coding during the content delivery, minimizing the network traffic. In this paper, we demonstrate that such coding opportunities can also be utilized to significantly reduce the communication load of wireless distributed computing applications. However, the proposed coded framework for wireless distributed computing differs significantly from the coded caching problems, mainly in the follow aspects.
\begin{itemize}
\item In~\cite{maddah2014fundamental, maddah2013decentralized}, a central server has the entire dataset and broadcasts coded messages to satisfy users' demands. In this work, the access point neither stores any part of the dataset, nor performs any computation. We designed new codes at both the users and the access point for data shuffling. 
\item The cache contents are placed without knowing the users' demands in the coded caching problems, while here the dataset placement is performed knowing that each user has her own unique computation request (input).
\item Our scheme is faced with the challenge of symmetric computation enforced by the MapReduce-type structure, i.e., a Map function computes intermediate values for all inputs. Such symmetry is not enforced in coded caching problems. 
 \end{itemize}

Other than reducing the communication load, a recent work~\cite{lee2015speeding} has also proposed to use codes to deal with the stragglers for a specific class of distributed computing jobs (e.g., matrix multiplication), and the optimal assignment of the coded tasks in a heterogeneous computing environment was addressed in~\cite{reisizadehmobarakeh2017coded}. Additionally, for these computing jobs, a unified coding framework was recently proposed in \cite{LMA16_unify} to achieve a tradeoff between the latency of computation and the load of communication, on which the scheme in \cite{li2016fundamental} and the scheme in \cite{lee2015speeding} achieve the two end points, minimizing the communication load and the computation latency respectively.

There have also been several recent works on communication design and resource allocation for mobile-edge computation offloading (see e.g., \cite{barbarossa2014communicating,khalili2016inter}), in which a part of the computation is offloaded to clouds located at the edges of cellular networks. In this scenario, recent works \cite{li2017codingfog,LMA_ISIT17} have proposed to exploit coding in edge processing to reduce the load of computation and improve the spectral efficiency. In contrast to the computation offloading model, in this paper, our focus is on the scenario that the ``edge'' only facilitates the communication required for distributed computing, and all computations are done \emph{distributedly} at the users.

\section{System Model}\label{sec:def}
	We consider a system that has $K$ mobile users, for some $K \in \mathbb{N}$. As illustrated in Fig.~\ref{fig:network}, all users are connected wirelessly to an access point (e.g., a cellular base station or a Wi-Fi router). The uplink channels of the $K$ users towards the access point are orthogonal to each other, and the signals transmitted by the access point on the downlink are received by all the users.
	
	\begin{figure}[htbp]
		\centering
		\includegraphics[width=0.48\textwidth]{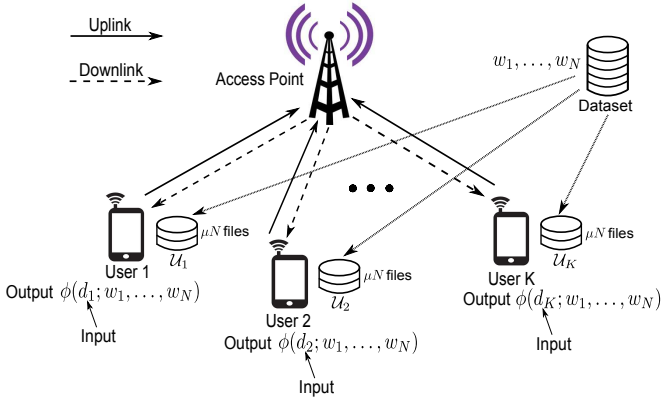}
		\caption{A wireless distributed computing system.}
		\label{fig:network}
	\end{figure}
		
	The system has a dataset (e.g., a feature repository of objects in a image recognition application) that is evenly partitioned into $N$ files $w_1,\ldots,w_N \in \mathbb{F}_{2^F}$, for some $N,F \in \mathbb{N}$. Every User~$k$ has a length-$D$ input $d_k \in \mathbb{F}_{2^D}$ (e.g., user's image in the image recognition application) to process using the $N$ files. To do that, as shown in Fig.~\ref{fig:network}, User~$k$ needs to compute 
	\begin{equation}\label{eq:output}
	\phi(\underbrace{d_k}_{\text{input}};\underbrace{w_1,\ldots,w_N}_{\text{dataset}}),
	\end{equation}
where $\phi: \mathbb{F}_{2^D} \times (\mathbb{F}_{2^F})^N \rightarrow \mathbb{F}_{2^B}$ is an output function that maps the input $d_k$ to an output result (e.g., the returned result after processing the image) of length $B \in \mathbb{N}$.

We assume that every mobile user has a local memory that can store up to $\mu$ fractions of the dataset (i.e., $\mu N$ files), for some constant parameter $\mu$ that does not scale with the number of users $K$. Throughout the  paper, we consider the case $\frac{1}{K} \leq \mu < 1$, such that each user does not have enough storage for the entire dataset, but the entire dataset can be stored collectively across all the users. We denote the set of indices of the files stored by User~$k$ as ${\cal U}_k$.
The selections of ${\cal U}_k$s are design parameters, and we denote the design of ${\cal U}_1,\ldots,{\cal U}_K$ as \emph{dataset placement}. The dataset placement is performed in prior to the computation (e.g., users download parts of the feature repository when installing the image recognition application). 
				
\begin{remark}
The employed physical-layer network model is rather simple and one can do better using a more detailed model and more advanced techniques. However we note that any wireless medium can be converted to our simple model using (1) TDMA on uplink; and (2) broadcast at the rate of weakest user on downlink. Since the goal of the paper is to introduce a ``coded'' framework for scalable wireless distributed computing, we decide to abstract out the physical layer and focus on the amount of data needed to be communicated.  $\hfill \square$
\end{remark}
												
\noindent {\bf Distributed Computing Model.} Motivated by prevalent distributed computing structures like MapReduce~\cite{dean2004mapreduce} and Spark~\cite{zaharia2010spark}, we assume that the computation for input $d_k$ can be decomposed as
\begin{equation}\label{eq:decom}
\phi(d_k;w_1,\ldots,w_N) = h(g_1(d_k;w_1),\ldots,g_N(d_k;w_N)),
\end{equation}
where as illustrated in Fig.~\ref{fig:frame},
\begin{itemize}
\item The ``Map'' functions $g_n(d_k;w_n): \mathbb{F}_{2^D} \times \mathbb{F}_{2^F} \!\rightarrow \! \mathbb{F}_{2^T}$, $n \in \{1,\ldots,N\}$, $k \in \{1,\ldots,K\}$, maps the input $d_k$ and the file $w_n$ into an \emph{intermediate value} $v_{k,n}\!=\!g_n(d_k;w_n)\!\in \! \mathbb{F}_{2^T}$, for some $T \!\in\! \mathbb{N}$,
\item The ``Reduce'' function $h: (\mathbb{F}_{2^T})^N \!\rightarrow \! \mathbb{F}_{2^B}$ maps the intermediate values for input $d_k$ in all files into the output value $\phi(d_k;w_1,\ldots,w_N)=h(v_{k,1},\ldots,v_{k,N})$, for all $k \in \{1,\ldots,K\}$.
\end{itemize}
	
\begin{remark}
	Note that for every set of output functions such a Map-Reduce decomposition exists (e.g., setting $g_n$$'s$ to identity and $h$ to $\phi(d_k;*)$). However, such a decomposition is not unique, and in the distributed computing literature, there has been quite some work on developing appropriate decompositions of computations like join, sorting and matrix multiplication (see e.g., \cite{dean2004mapreduce,rajaraman2011mining}), which are suitable for efficient distributed computing. Here we do not impose any constraint on how the Map and Reduce functions are chosen (for example, they can be arbitrary \emph{linear} or \emph{non-linear} functions).  $\hfill \square$
	\end{remark}
	
	\begin{figure}[htbp]
		\centering
		\includegraphics[width=0.48\textwidth]{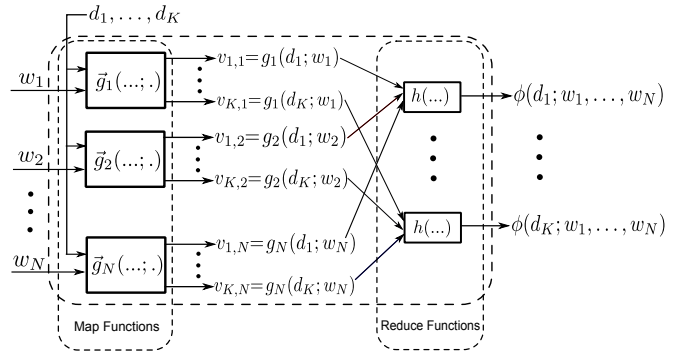}
		\caption{A two-stage distributed computing framework decomposed into Map and Reduce functions.}
		\label{fig:frame}
	\end{figure}
	
We focus on the applications in which the size of the users' inputs is much smaller than the size of the computed intermediate values, i.e., $D \ll T$. As a result, the overhead of disseminating the inputs is negligible, and we assume that the users' inputs $d_1,\ldots,d_K$ are known at each user before the computation starts.
	
\begin{remark}
The above assumption holds for various wireless distributed computing applications. For example, in a mobile navigation application, an input is simply the addresses of the two end locations. The computed intermediate results contain all possible routes between the two end locations, from which the shortest one (or the fastest one considering the traffic condition) is computed for the user. Similarly, for a set of ``filetring'' applications like the aforementioned image recognition (or similarly augmented reality) and recommendation systems, the inputs are light-weight queries (e.g., the feature vector of an image) that are much smaller than the filtered intermediate results containing all attributes of related information. For example, an input can be multiple words describing the type of restaurant a user is interested in, and the intermediate results returned by a recommendation system application can be a list of relevant information that include customers' comments, pictures, and videos of the recommended restaurants.
$\hfill \square$
\end{remark}

Following the decomposition in (\ref{eq:decom}), the overall computation proceeds in three phases: \emph{Map}, \emph{Shuffle}, and \emph{Reduce}.
	
	\noindent {\bf Map Phase:} 
	User $k$, $k\in \{1,\ldots,K\}$ computes the Map functions of $d_1,\ldots,d_K$ based on the files in $\mathcal{U}_k$. For each input $d_k$ and each file $w_n$ in $\mathcal{U}_k$, User $k$ computes $g_n(d_k,w_n)=v_{k,n}$. 
			
	\noindent {\bf Shuffle Phase:} 
	In order to compute the output value for the input $d_k$, User $k$ needs the intermediate values that are \emph{not} computed \emph{locally} in the Map phase, i.e., $\{v_{k,n}: n \notin \mathcal{U}_k\}$.Users exchange the needed intermediate values via the access point they all wirelessly connect to. As a result, the Shuffle phase breaks into two sub-phases: \emph{uplink communication} and \emph{downlink communication}.
	
	On the uplink, user $k$ creates a message $W_k$ as a function of the intermediate values computed locally, i.e., $W_k = \psi_k\left(\{v_{k,n}:k \in \{1,\ldots,K\}, n \in \mathcal{U}_k\}\right)$, and communicates $W_k$ to the access point. 
	
	\begin{definition}[Uplink Communication Load]
	We define the \emph{uplink communication load}, denoted by $L_{u}$, as the total number of bits in all uplink messages $W_1,\ldots,W_K$, normalized by the number of bits in the $N$ intermediate values required by a user (i.e., $NT$). 
	\end{definition}
	
	We assume that the access point does not have access to the dataset. Upon decoding all the uplink messages $W_{1},\ldots,W_{K}$, the access point generates a message $X$ from the decoded uplink messages, i.e., $X = \rho(W_{1},\ldots,W_{K})$, then broadcasts $X$ to all users on the downlink.
	
	\begin{definition}[Downlink Communication Load]
		We define the \emph{downlink communication load}, denoted by $L_d$, as the number of bits in the downlink message $X$, normalized by $NT$.   
	\end{definition}
	
	\noindent {\bf Reduce Phase:} 
	User $k$, $k\in\{1,\ldots,K\}$ uses the locally computed results $\{\vec{g}_n: n\in \mathcal{U}_k\}$ and the decoded downlink message $X$ to construct the inputs to the corresponding Reduce function, and calculates the output value $\phi(d_k;w_1,\ldots,w_N)=h(v_{k,1},\ldots, v_{k,N})$. 
	
{\bf Example (Uncoded Scheme).} As a benchmark, we consider an uncoded scheme, where each user receives the needed intermediate values sent uncodedly by some other users and forwarded by the access point, achieving the communication loads $L_u^{\textup{uncoded}}(\mu)  = L_d^{\textup{uncoded}}(\mu) = \mu K \cdot (\tfrac{1}{\mu}-1)$.
 
We note that the above communication loads of the uncoded scheme grow with the number of users $K$, overwhelming the limited spectral resources. In this paper, we argue that by utilizing coding at the users and the access point, we can accommodate any number of users with a \emph{constant} communication load. Particularly, we propose in the next section a scalable coded wireless distributed computing (CWDC) scheme that achieves the minimum possible uplink and downlink communication load simultaneously, i.e., 
\begin{align}
L_u^{\textup{coded}} &= L_u^{\textup{optimum}} \approx \tfrac{1}{\mu}-1, \\
L_d^{\textup{coded}} &= L_d^{\textup{optimum}} \approx  \tfrac{1}{\mu}-1.
\end{align}

\section{The Proposed CWDC Scheme}\label{sec:scheme}
In this section, we present the proposed CWDC scheme for a \emph{centralized} setting, in which the dataset placement is designed in a centralized manner knowing the number and the identities of the users that will participate in the computation. We first consider the storage size $\mu  \in \{\frac{1}{K},\frac{2}{K},\ldots,1\}$ such that $\mu K \in \mathbb{N}$. We assume that $N$ is sufficiently large such that $N = {K \choose \mu K}\eta$ for some $\eta \in \mathbb{N}$. 

\noindent {\bf Dataset Placement and Map Phase Execution.} We evenly partition the indices of the $N$ files into ${K \choose \mu K}$ disjoint batches, each containing the indices of $\eta$ files. We denote a batch of file indices as ${\cal B}_{\cal T}$, which is labelled by a unique subset $\mathcal{T} \subset \{1,\ldots,K\}$ of size $|{\cal T}|=\mu K$. As such defined, we have
\begin{equation}
\{1,\ldots,N\} \!=\! \{i: i \in \mathcal{B}_{\cal T}, {\cal T} \subset \{1,\ldots,K\}, |{\cal T}|=\mu K \}.
\end{equation}
User~$k$, $k \in \{1,\ldots,K\}$, stores locally all the files whose indices are in $\mathcal{B}_{\cal T}$ if $k \in \mathcal{T}$. That is, 
\begin{align}
{\cal U}_k = \underset{{\cal T}: |{\cal T}|=\mu K, k \in {\cal T}}{\cup} {\cal B}_{\cal T}.
\end{align}
As a result, each of the $N$ files is stored by $\mu K$ distinct users.
After the Map phase, User $k$, $k \in \{1,\ldots,K\}$, knows the intermediate values of all $K$ output functions in each file whose index is in $\mathcal{U}_k$, i.e., $\{v_{q,n}:q \in \{1,\ldots,K\}, n \in \mathcal{U}_k\}$.

In the example in Section~\ref{sec:example}, the indices of the 6 files are partitioned into ${3 \choose 2}=3$ batches, each containing the indices of 2 files. Each user stores the files whose indices are in 2 out of the 3 batches. Hence, each user stores a total of 4 files.

\noindent {\bf Uplink Communication.} For any subset ${\cal W}\subset \{1,\ldots,K\}$, and any $k \notin {\cal W}$, we denote the set of intermediate values needed by User $k$ and known \emph{exclusively} by users in $\mathcal{W}$ as $\mathcal{V}_{\mathcal{W}}^{k}$. More formally:
\begin{equation}\label{eq:V}
\mathcal{V}_{\mathcal{W}}^{k} \triangleq \{v_{k,n}: n \in \underset{i \in {\cal W}}{\cap} \mathcal{U}_i, n \notin \underset{i \notin {\cal W}}{\cup} \mathcal{U}_i\}.
\end{equation}

In the example in Section~\ref{sec:example}, we have $\mathcal{V}_{\{2,3\}}^{1}=\{v_{1,5},v_{1,6}\}$, $\mathcal{V}_{\{1,3\}}^{2}=\{v_{2,1},v_{2,2}\}$ and $\mathcal{V}_{\{1,2\}}^{3}=\{v_{3,3},v_{3,4}\}$.

For all subsets $\mathcal{S} \subseteq \{1,\ldots,K\}$ of size $\mu K+1$:
\begin{enumerate}
\item For each User $k \in \mathcal{S}$, $\mathcal{V}_{\mathcal{S}\backslash \{k\}}^{k}$ is the set of intermediate values that are requested by User~$k$ and are in the files whose indices are in the batch ${\cal B}_{\mathcal{S}\backslash \{k\}}$, and they are exclusively known at all users whose indices are in $\mathcal{S}\backslash \{k\}$. We evenly and arbitrarily split $\mathcal{V}_{\mathcal{S}\backslash \{k\}}^{k}$, into $\mu K$ disjoint segments $\{\mathcal{V}_{\mathcal{S} \backslash \{k\},i}^{k}\!:\! i \in {\cal S} \backslash \{k\}\}$, where $\mathcal{V}_{\mathcal{S} \backslash \{k\},i}^{k}$ denotes the segment associated with User~$i$ in ${\cal S} \backslash \{k\}$ for User~$k$. That is, $\mathcal{V}^{k}_{\mathcal{S}\backslash \{k\}} \!=\! \underset{i \in {\cal S} \backslash \{k\}}{\cup} \mathcal{V}_{\mathcal{S} \backslash \{k\},i}^{k}$.
\item User $i$, $i \in \mathcal{S}$, sends the bit-wise XOR, denoted by $\oplus$, of all the segments associated with it in ${\cal S}$, i.e., User $i$ sends the coded segment $W_i^{\cal S} \triangleq \underset{k \in \mathcal{S} \backslash \{i\}}\oplus \mathcal{V}^{k}_{\mathcal{S}\backslash \{k\},i}$.
\end{enumerate}

Since the coded message $W_i^{\cal S}$ contains $\frac{\eta}{\mu K} T$\footnote{Here we assume that $T$ is sufficiently large such that $\frac{T}{\mu K} \in \mathbb{N}$.} bits for all $i \in {\cal S}$, there are a total of $\frac{(\mu K+1)\eta}{\mu K}T$ bits communicated on the uplink in every subset ${\cal S}$ of size $\mu K+1$. Therefore, the uplink communication load achieved by this coded scheme is $L_u^{\textup{coded}}(\mu) = \tfrac{{K\choose \mu K+1} (\mu K+1) \cdot \eta \cdot T}{\mu K  \cdot NT} =\tfrac{1}{\mu}-1, \; \mu \in \{\tfrac{1}{K},\tfrac{2}{K},\ldots,1\}$.

\noindent {\bf Downlink Communication.} For each subset $\mathcal{S} \subseteq \{1,\ldots,K\}$ of size $\mu K+1$, and ${\cal S} = \{i_1,i_2,\ldots,i_{\mu K+1}\}$, the access point computes $\mu K$ random linear combinations of the uplink messages generated based on the subset ${\cal S}$: $C^{\cal S}_j(W_{i_1}^{\cal S},W_{i_2}^{\cal S},\ldots, W_{i_{\mu K+1}}^{\cal S}), \; j = 1,\ldots, \mu K$, and multicasts them to all users in ${\cal S}$.

Since each linear combination contains $\frac{\eta}{\mu K} T$ bits, the coded scheme achieves a downlink communication load $L_d^{\textup{coded}}(\mu) \!=\! \tfrac{{K\choose \mu K+1} \eta \cdot T}{NT} \!=\! \tfrac{\mu K}{\mu K+1} \!\cdot\! (\tfrac{1}{\mu}-1), \, \mu \in \{\tfrac{1}{K},\tfrac{2}{K},\ldots,1\}$.

After receiving the random linear combinations $C^{\cal S}_1,\ldots,C^{\cal S}_{\mu K}$, User $i$, $i \in {\cal S}$, cancels all segments she knows locally, i.e., $\underset{k \in {\cal S} \backslash \{i\}}{\cup}\{\mathcal{V}_{\mathcal{S} \backslash \{k\},j}^{k}: j \in \mathcal{S} \backslash \{k\}\}$. Consequently, User~$i$ obtains $\mu K$ random linear combinations of the required $\mu K$ segments $\{\mathcal{V}^{i}_{\mathcal{S}\backslash \{i\},j}: j \in {\cal S}\backslash \{i\}\}$.

\begin{remark}
The above uplink and downlink communication schemes require coding at both the users and the access point, creating multicasting messages that are simultaneously useful for many users. Such idea of efficiently creating and exploiting \emph{coded multicast} opportunities was initially proposed in the coded caching problems in~\cite{maddah2014fundamental, maddah2013decentralized}, and extended to D2D networks in~\cite{ji2014fundamental}. While simply forwarding the coded uplink packets on the downlink can already reduce the downlink communication load by a factor of $\mu K$, performing random linear coding at the access point achieves a higher reduction factor of $\mu K+1$. We note that this type of random linear coding at the access point has been utilized before in solving network coding problems (see, e.g.,~\cite{ahlswede2000network,KM03,HKMKE03}) and bi-directional relaying problems (see, e.g., \cite{KMT-bi,AST-two}).
$\hfill \square$
\end{remark}

When $\mu K$ is not an integer, we can first expand $\mu = \alpha \mu_1 + (1-\alpha)\mu_2$ as a convex combination of $\mu_1 \triangleq \lfloor \mu K \rfloor/K$ and $\mu_2 \triangleq \lceil \mu K \rceil/K$. Then we partition the set of the $N$ files into two disjoint subsets $\mathcal{I}_1$ and $\mathcal{I}_2$ of sizes $|\mathcal{I}_1| = \alpha N$ and $|\mathcal{I}_2| = (1-\alpha) N$. We next apply the above coded scheme respectively to the files in $\mathcal{I}_1$ where each file is stored at $\mu_1K$ users, and the files in $\mathcal{I}_2$ where each file is stored at $\mu_2K$ users, yielding the following communication loads.
\begin{align}
L_u^{\textup{coded}}(\mu) &= \alpha (\tfrac{1}{\mu_1}-1)+ (1-\alpha) (\tfrac{1}{\mu_2}-1),\\
L_d^{\textup{coded}}(\mu) &=\alpha \tfrac{\mu_1 K}{\mu_1 K+1} \cdot (\tfrac{1}{\mu_1}-1) + (1-\alpha)\tfrac{\mu_2 K}{\mu_2 K+1} \cdot (\tfrac{1}{\mu_2}-1).
\end{align} 

Hence, for general storage size $\mu$, CWDC achieves the following communication loads.
	\begin{align}
	L_u^{\textup{coded}}(\mu)&=\textup{Conv} (\tfrac{1}{\mu}-1), \label{eq:allR-up}\\
	L_d^{\textup{coded}}(\mu)&=\textup{Conv} (\tfrac{\mu K}{\mu K+1}\cdot(\tfrac{1}{\mu}-1)), \label{eq:allR-down}
	\end{align}
where $\textup{Conv}(f(\mu))$ denotes the lower convex envelope of the points $\{(\mu,f(\mu))\!:\!\mu \in\{\frac{1}{K},\frac{2}{K},...,1\}\}$. 

We summarize the performance of the proposed CWDC scheme in the following theorem.
	\begin{theorem}
	For a wireless distributed computing application with a dataset of $N$ files, and $K$ users that each can store $\mu \in \{\frac{1}{K},\frac{2}{K},\ldots,1\}$ fraction of the files, the proposed CWDC scheme achieves the following uplink and downlink communication loads for sufficiently large $N$. 
	\begin{align}
	L_u^{\textup{coded}}(\mu)&=\tfrac{1}{\mu}-1, \\
	 L_d^{\textup{coded}}(\mu)&=\tfrac{\mu K}{\mu K+1}\cdot(\tfrac{1}{\mu}-1).
	\end{align}
For general $\frac{1}{K} \leq \mu \leq 1$, the achieved loads are as stated in (\ref{eq:allR-up}) and (\ref{eq:allR-down}).
	\end{theorem}
	
\begin{remark}\label{scale}
	Theorem~1 implies that, for large $K$, $L_u^{\textup{coded}}(\mu) \approx L_d^{\textup{coded}}(\mu) \approx \tfrac{1}{\mu}-1$, which is independent of the number of users. Hence, we can accommodate any number of users without incurring extra communication load, and the proposed scheme is \emph{scalable}. The reason for this phenomenon is that, as more users joint the network, with an appropriate dataset placement, we can create coded multicasting opportunities to reduce the communication loads by a factor of $\mu K$, which is the size of the aggregated memory of all users in the system, and scales linearly with $K$ ($\mu$ is a constant). Such phenomenon was also observed in the context of cache networks (see e.g.,~\cite{maddah2014fundamental}). $\hfill \square$
\end{remark}
	
\begin{remark}
As illustrated in Fig.~\ref{fig:load-storage} in Section~\ref{sec:intro}, compared with the uncoded scheme, the proposed CWDC scheme utilizes coding at the mobile users and the access point to reduce the uplink and downlink communication load by a factor of $\mu K$ and $\mu K+1$ respectively, which scale linearly with the number of users $K$. When $\mu=\frac{1}{K}$, which is the minimum storage size required to accomplish distributed computing, the CWDC scheme reduces to the uncoded scheme when the access point simply forwards the received uncoded packets. $\hfill \square$
\end{remark}

\begin{remark}
Compared with distributed computing over wired servers where we only need to design one data shuffling scheme between servers in~\cite{li2016fundamental}, here in the wireless setting we jointly design uplink and downlink shuffling schemes, which minimize both the uplink and downlink communication loads.    $\hfill \square$
\end{remark}

\begin{remark}
We can view the Shuffle phase as an instance of the index coding problem~\cite{birk2006coding,bar2011index}, in which a central server aims to design a broadcast message with minimum length to satisfy the requests of all the clients, given the clients' local side information. While a random linear network coding approach (see e.g., \cite{ahlswede2000network,KM03,HKMKE03}) is sufficient to implement any multicast communication, it is generally sub-optimal for index coding problems where every client requests different messages. However, for the considered wireless distributed computing scenario where we are given the flexibility of designing dataset placement (thus the side information), we can prove that the proposed CWDC scheme is optimum in minimizing communication loads (see Section~\ref{sec:opt}).  $\hfill \square$
\end{remark}

\begin{remark}
We note that the coding opportunities created and exploited in the proposed coded scheme belong to a type of in-network coding, which aims to combat interference in wireless networks, and deliver the information bits required by each of the users respectively with maximum spectral efficiency. This type of coding is distinct from source coding, or data compression (see e.g.,~\cite{welton2011improving}), which aims to remove the redundant information in the original intermediate values each of the users requests. Interestingly, the above proposed coded communication scheme can be applied on top of data compression. That is, we can first compress the intermediate values to minimize the number of information bits each user requests, then we apply the proposed coded communication scheme on the compressed values, in order to deliver them to intended users with minimum utilization of the wireless links.  $\hfill \square$
\end{remark}

So far, we have considered the scenario where the dataset placement is designed in a centralized manner, i.e., the dataset placement is designed knowing which users will use the application. However, a more practical scenario is that before computation, the dataset placement at each user is performed in a decentralized manner without knowing when the computation will take place and who will take part in the computation. In the next section, we describe how we can extend the proposed CWDC scheme to facilitate the computation in such a decentralized setting.

\section{The Proposed CWDC Scheme for the Decentralized Setting}\label{sec:decent}

We consider a decentralized system, in which a random and a priori unknown subset of users, denoted by ${\cal K}$, participate in the computation. The dataset placement is performed \emph{independently} at each user by \emph{randomly} storing a subset of $\mu N$ files, according to a common placement distribution $P$. In this case, we define the \emph{information loss} of the system, denoted by $\Delta$, as the fraction of the files that are not stored by any participating user. 

Once the computation starts, the participating users in ${\cal K}$ of size $K$ are fixed, and their identities are revealed to all the participating users. Then they collaboratively perform the computation as in the centralized setting. The participating users process their inputs over the available part of the dataset stored collectively by all participating users. More specifically, every user~$k$ in ${\cal K}$ now computes  	
\begin{equation}\label{eq:part}
\phi(\underbrace{d_k}_{\text{input}};\underbrace{\{w_n: n \in \underset{k \in {\cal K}}{\cup}{\cal U}_k\}}_{\text{available dataset}}).
\end{equation}

In what follows, we present the proposed CWDC scheme for the above decentralized setting, including a random dataset placement strategy, an uplink communication scheme and a downlink communication scheme.

\noindent {\bf Dataset Placement.} We use a \textit{uniformly random} dataset placement, in which every user independently stores $\mu N$ files uniformly at random. With high probability for large $N$, the information loss approximately equals $(1-\mu)^{K}$, which converges quickly to $0$ as $K$ increases.

For a decentralized random dataset placement, files are stored by random subsets of users. During data shuffling, we first \emph{greedily} categorize the available files based on the number of users that store the file, then for each category we deliver the corresponding intermediate values in an \emph{opportunistic} way using the coded communication schemes described in Section~\ref{sec:scheme} for the centralized setting.

\noindent {\bf Uplink Communication.}
For all subsets $\mathcal{S} \subseteq \{1,\ldots,K\}$ with size $|\mathcal{S}|\geq 2$:
\begin{enumerate}
	\item For each $k \in \mathcal{S}$, we evenly and arbitrarily split $\mathcal{V}_{\mathcal{S}\backslash \{k\}}^{k}$ defined in (\ref{eq:V}), into $|\mathcal{S}|\!-\!1$ disjoint segments $\mathcal{V}^{k}_{\mathcal{S}\backslash \{k\}} = \{ \mathcal{V}_{\mathcal{S} \backslash \{k\},i}^{k}: i \in {\cal S} \backslash \{k\}\}$, and associate the segment $\mathcal{V}_{\mathcal{S} \backslash \{k\},i}^{k}$ with the user $i \in {\cal S} \backslash \{k\}$. 
	\item User $i$, $i \in \mathcal{S}$, sends the bit-wise XOR, denoted by $\oplus$, of all the segments associated with it in ${\cal S}$, i.e., User $i$ sends the coded segment $W_i^{\cal S} \triangleq \underset{k \in \mathcal{S} \backslash \{i\}}\oplus \mathcal{V}^{k}_{\mathcal{S}\backslash \{k\},i}$. \footnote{Since the dataset placement is now randomized, we zero-pad all elements in $\{\mathcal{V}^{k}_{\mathcal{S}\backslash \{k\},i} :k \in {\cal S}\backslash \{i\}\}$ to the maximum length $\underset{k \in {\cal S}\backslash\{i\}}{\max}|\mathcal{V}_{\mathcal{S} \backslash \{k\},i}^k|$ in order to complete the XOR operation.}
\end{enumerate}

Using the proposed uniformly random dataset placement, for any subset $\mathcal{S}\subseteq \{1,...,K\}$, the number of files exclusively stored by all users in $\mathcal{S}$ can be characterized by $\mu^{|\mathcal{S}|}(1-\mu)^{K-|\mathcal{S}|}  N + o(N)$.

Thus, when the proposed communication scheme proceeds on a subset ${\cal S}$ of size $|{\cal S}|=j+1$ users, the resulting uplink communication load converges to $\frac{j+1}{j}\mu^{j}(1-\mu)^{K-j}$ for large $N$.

\noindent {\bf Downlink Communication.}
For all $\mathcal{S} \subseteq \{1,\ldots,K\}$ of size $|\mathcal{S}|\geq 2$, the access point computes $|\mathcal{S}|-1$ random linear combinations of the uplink messages generated based on the subset ${\cal S}$: $C^{\cal S}_j(\{W_i^{\cal S}: i \in {\cal S}\}), \; j = 1,\ldots, |\mathcal{S}|-1$, and multicasts them to all users in ${\cal S}$.

We summarize the performance of the proposed decentralized CWDC scheme in the following theorem.

\begin{theorem}
		For an application with a dataset of $N$ files, and $K$ users that each can store $\mu$ fraction of the files, the proposed decentralized CWDC scheme achieves an information loss $\Delta = \left(1-\mu\right)^K$ and the following communication loads with high probability for sufficiently large $N$.
		 \begin{align}
		 L_{\textup{decent},u}^{\textup{coded}} &= \sum_{j=1}^{K-1} \binom{K}{j+1}\frac{j+1}{j} \mu ^j \left(1-\mu\right)^{K-j},\\
		 L_{\textup{decent},d}^{\textup{coded}} &= \sum_{j=1}^{K-1} {K \choose j+1}\mu^j \left(1-\mu\right)^{K-j}. \label{eq:loadD}
		 \end{align}
\end{theorem}

\begin{figure}[htbp]
	\centering
	\subfigure[Uplink.]{\includegraphics[width=0.3\textwidth]{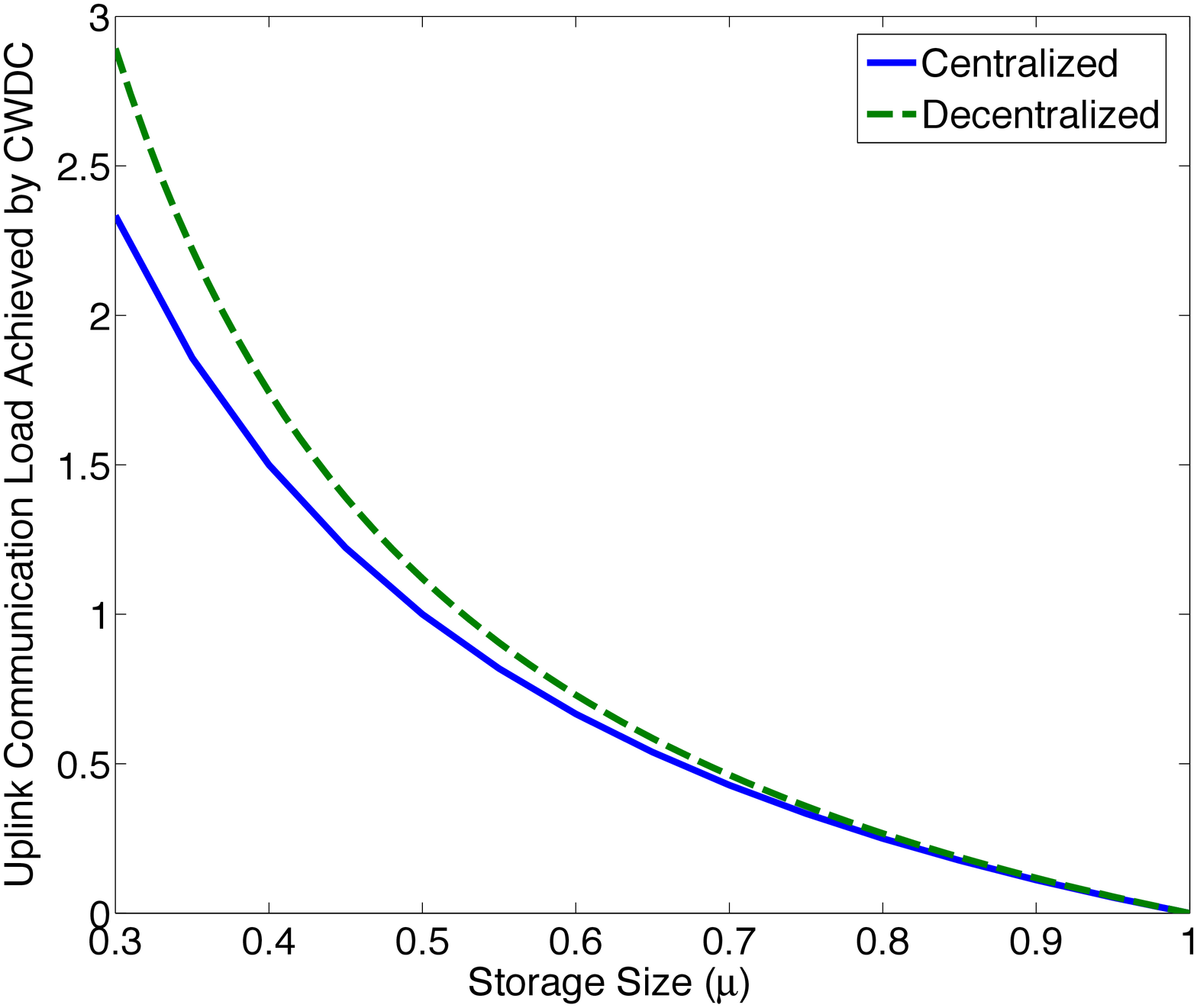}
		\label{fig:up}}
	\subfigure[Downlink.]{\includegraphics[width=0.3\textwidth]{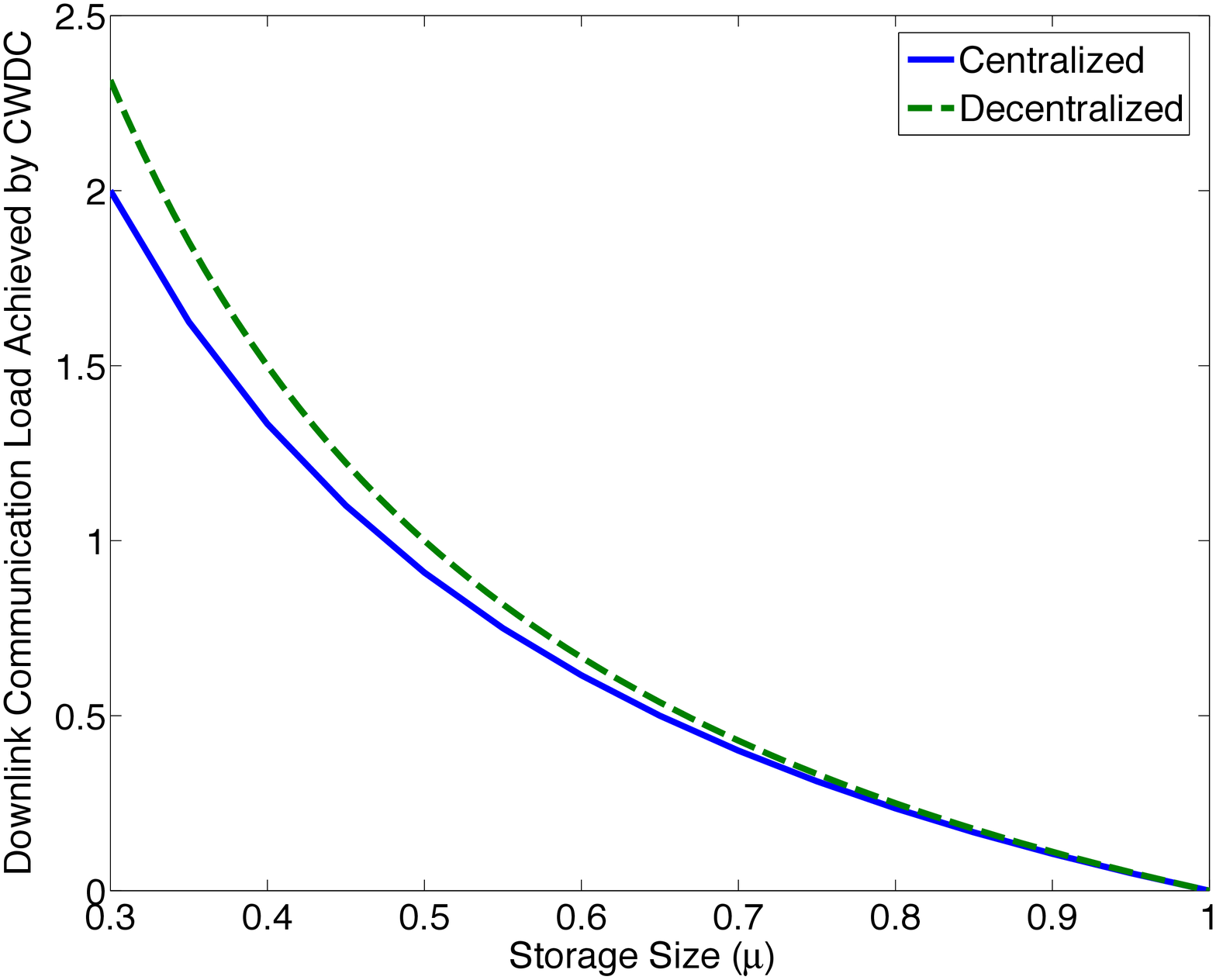}
		\label{fig:down}}
	\caption{Comparison of the communication loads achieved by the centralized and the decentralized CWDC schemes, for a network of $K=20$ participating users.}
	\label{fig:load}
\end{figure}

\begin{remark}\label{decastm}
	In Fig.~\ref{fig:load}, we numerically evaluate the communication loads achieved by the proposed centralized and decentralized schemes, in a network with 20 participating users. We observe that although the loads of the decentralized scheme are higher than those of the centralized scheme, the communication performances under these two settings are very close to each other. As $K$ becomes large, the information loss achieved by the decentralized CWDC approaches $0$, and both loads in (\ref{eq:loadD}) approach $ \frac{1}{\mu}-1$, which equals the asymptotic loads achieved by the centralized scheme (see Remark~\ref{scale}). Hence, when the number of participating users is large, there is little loss in making the system decentralized.
$\hfill \square$		
\end{remark}

\begin{figure}[htbp]
		\centering
		\includegraphics[width=0.4\textwidth]{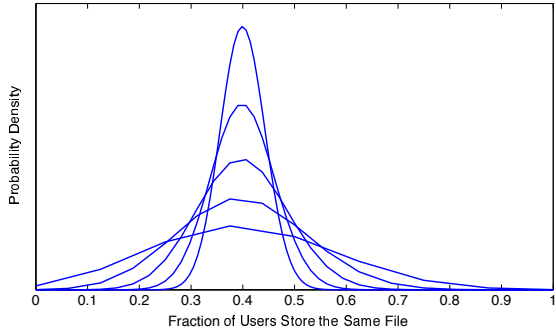}
		\caption{Concentration of the number of users each files is stored at around $\mu K$. Each curve demonstrates the normalized fraction of files that are stored by different numbers of users, for a particular number of participating users $K$. The density functions are computed for a storage size $\mu=0.4$, and for $K=2^3,...,2^7$.}
		\label{fig:dist}
\end{figure}

\begin{remark}
	To understand the fact that the proposed decentralized scheme performs close to the centralized one when the number of participating users is large, we notice the fact that when uniformly random dataset placement is used, as demonstrated in Fig.~\ref{fig:dist}, almost all files are stored by approximately $\mu K$ users for large $K$, which coincides with the optimal dataset placement for the centralized setting. Thus, coding gains of the the proposed decentralized communication schemes would be also very close to those of the centralized schemes. Such phenomenon was also observed in~\cite{maddah2013decentralized} for caching problems with decentralized content placement. $\hfill \square$	
\end{remark}

\section{Optimality of the Proposed CWDC Schemes}\label{sec:opt}
In this section, we demonstrate in the following two theorems, that the proposed CWDC schemes achieve the minimum uplink and downlink communication loads using any scheme, for the centralized setting and the decentralized setting respectively. 

\begin{theorem}
	For a centralized wireless distributed computing application using any dataset placement and communication schemes that achieve an uplink load $L_u$ and a downlink load $L_d$, $L_u$ and $L_d$ are lower bounded by $L^{\textup{coded}}_u(\mu)$ and $L^{\textup{coded}}_d(\mu)$ as stated in Theorem~1 respectively.
\end{theorem}

\begin{remark}
Using Theorem~1 and~3, we have completely characterized the minimum achievable uplink and downlink communication loads, using \emph{any} dataset placement, uplink and downlink communication schemes for the centralized setting. This implies that the proposed centralized CWDC scheme \emph{simultaneously} minimizes both uplink and downlink communication loads required to accomplish distributed computing, and no other scheme can improve upon it. This also demonstrates that there is no fundamental tension between optimizing uplink and downlink communication in wireless distributed computing. $\hfill \square$
\end{remark}

For a dataset placement $\mathcal{U} = \{\mathcal{U}_k\}_{k=1}^{K}$, we denote the minimum possible uplink and downlink communication loads, achieved by any uplink-downlink communication scheme to accomplish wireless distributed computing, by $L_u^*(\mathcal{U})$ and $L_d^*(\mathcal{U})$ respectively. We next prove Theorem~3 by deriving lower bounds on $L_u^*(\mathcal{U})$ and $L_d^*(\mathcal{U})$ respectively.
	
\subsection{Lower Bound on $L_u^*(\cal U)$}
	For a given dataset placement ${\cal U}$, we denote the number of files that are stored at $j$ users as $a^j_{{\cal U}}$, for all $j \in \{1,\ldots,K\}$, i.e.,
	\begin{equation} \label{eq:count}
	a^j_{{\cal U}} = \sum \limits_{{\cal J} \subseteq \{1,\ldots,K\}: |{\cal J}|=j} |(\underset{k \in {\cal J}}{\cap} {\cal U}_k) \backslash (\underset{i \notin {\cal J}}{\cup} {\cal U}_i )|.
	\end{equation}

For any ${\cal U}$, it is clear that $\{a^j_{\cal U}\}_{j=1}^K$ satisfy
	\begin{align}
	\sum_{j=1}^{K} a^j_{\cal U}&=N, \label{eq:convex} \\ 
	\sum_{j=1}^{K} j a^j_{\cal U} &=\mu NK. \label{eq:compute}
	\end{align}
	
We start the proof with the following lemma, which characterizes a lower bound on $L_u^*(\mathcal{U})$ in terms of the distribution of the files in the dataset placement ${\cal U}$, i.e., $a^1_{{\cal U}},\ldots,a^K_{{\cal U}}$.
\begin{lemma}
$L_u^*(\mathcal{U}) \geq \sum\limits_{j=1}^{K} \frac{a^{j}_{\cal U}}{N}\cdot\frac{K-j}{j}$.
\end{lemma} 

\begin{proof}[Proof Sketch]
We know from the model of the distributed computing system that given the local computation results, and the downlink broadcast message $X$, each mobile user should be able to recover all the required intermediate values for the local Reduce function. Since the downlink message $X$ is generated as a function of the uplink messages $W_1,\ldots,W_K$ at the access point, a user can of course recover the required intermediate values if she were given $W_1,\ldots,W_K$ instead of $X$. Having observed the above fact, we can then prove Lemma~1 following the similar steps in the proof of Lemma~1 in~\cite{li2016fundamental}, in which each user can broadcast her message, generated as a function of the local computation results, to all other users.
\end{proof}

Lemma 1 implies that in order to deliver an intermediate value of size $T$ bits that is known at $j$ users and needed by one of the remaining $K-j$ users, the $j$ users who know this value need to communicate at least $T/j$ bits on the uplink.
	
Next, since the function $\frac{K-j}{j}$ in Lemma~1 is convex in $j$, and by (\ref{eq:convex}) that $\sum\limits_{j=1}^{K} \frac{a^j_{\cal U}}{N} =1$ and (\ref{eq:compute}), we have
	\begin{align}\label{eq:general}
	L_u^*({\cal U})\geq  \tfrac{K-\sum\limits_{j=1}^K j\frac{a^j_{\cal U}}{N}}{\sum\limits_{j=1}^K j\frac{a^j_{\cal U}}{N}} =\tfrac{K- \mu K}{\mu K}= \tfrac{1}{\mu}-1.
	\end{align}
	
	We can further improve the lower bound in (\ref{eq:general}) for a particular $\mu$ such that $\mu K \notin \mathbb{N}$. For a given storage size $\mu$, we first find two points $(\mu_1, \frac{1}{\mu_1}-1)$ and $(\mu_2, \frac{1}{\mu_2}-1)$, where $\mu_1 \triangleq \lfloor \mu K \rfloor /K$ and $\mu_2 \triangleq \lceil \mu K \rceil /K$. Then we find the line $p+qt$ connecting these two points as a function of $t$, $\frac{1}{K} \leq t \leq 1$, for some constants $p,q\in \mathbb{R}$. We note that $p$ and $q$ are different for different $\mu$ and
	\begin{align}
	p+qt|_{t=\mu_1} &=  \frac{1}{\mu_1}-1,\\
	p+qt|_{t=\mu_2} &=  \frac{1}{\mu_2}-1.
	\end{align}

 Then by the convexity of the function $\frac{1}{t}-1$, the function $\frac{1}{t}-1$ cannot be smaller then the function $p + qt$ at the points $t = \frac{1}{K}, \frac{2}{K},\ldots,1$. That is, for all $t \in \{\frac{1}{K},\ldots,1\}$,
	\begin{equation}
	\frac{1}{t}-1 \geq p + qt.
	\end{equation}
	
By Lemma~1, we have 
	\begin{align}
	L_u^*({\cal U})&\geq \sum_{j=1}^{K} \frac{a^{j}_{\cal U}}{N} \cdot \frac{K-j}{j}\\
	&= \sum_{t =\frac{1}{K},\ldots,1} \frac{a^{tK}_{\cal U}}{N} \cdot \big(\tfrac{1}{t}-1\big)\\
	&\geq \sum_{t =\frac{1}{K},\ldots,1} \frac{a^{tK}_{\cal U}}{N} \cdot (p+qt)\\
	& = p+q\mu,
	\end{align}
	
	Therefore, for general $\frac{1}{K} \leq \mu \leq 1$, $L_u^*({\cal U})$ is lower bounded by the lower convex envelope of the points $\{(\mu,\frac{1}{\mu}-1): \mu \in\{\frac{1}{K},\frac{2}{K},...,1\}\}$. 
	
\subsection{Lower Bound on $L_d^*({\cal U})$}
	The lower bound on the minimum downlink communication load $L_d^*({\cal U})$ can be proved following the similar steps of lower bounding the minimum uplink communication load $L_u^*({\cal U})$, after making the following enhancements to the downlink communication system:
	\begin{itemize}
		\item We consider the access point as the $(K+1)$th user who has stored all $N$ files and has a virtual input to process. Thus the enhanced downlink communication system has $K+1$ users, and the dataset placement for the enhanced system
		\begin{equation}
		\bar{\cal U} \triangleq \{{\cal U}, {\cal U}_{K+1}\},
		\end{equation}
		where ${\cal U}_{K+1}$ is equal to $\{1,\ldots,N\}$.
		\item We assume that every one of the $K+1$ users can broadcast to the rest of the users, where the broadcast message is generated by mapping the locally stored files.
	\end{itemize}
	
	Apparently the minimum downlink communication load of the system cannot increase after the above enhancements. Thus the lower bound on the minimum downlink communication load of the enhanced system is also a lower bound for the original system.
	
	Then we can apply the same arguments in the proof of Lemma~1 to the enhanced downlink system of $K+1$ users, obtaining a lower bound on $L_d^*(\mathcal{U})$, as described in the following corollary:
	\begin{corollary}
		$L_d^*(\mathcal{U}) \geq \sum\limits_{j=1}^{K} \frac{a^{j}_{\cal U}}{N}\cdot\frac{K-j}{j+1}$.
	\end{corollary} 
	
	\begin{proof}
	Applying Lemma~1 to the enhanced downlink system yields
	\begin{align}
	L_d^*(\bar{\mathcal{U}}) &\geq \sum\limits_{j=1}^{K+1} \frac{a^{j}_{\bar{\cal U}}}{N}\cdot\frac{K+1-j}{j} 
	\geq \sum\limits_{j=2}^{K+1} \frac{a^{j}_{\bar{\cal U}}}{N}\cdot\frac{K+1-j}{j} \\ &=\sum\limits_{j=1}^{K} \frac{a^{j+1}_{\bar{\cal U}}}{N}\cdot\frac{K-j}{j+1}.\label{eq:enhance}
	\end{align}
	
	Since the access point has stored every file, $a^{j+1}_{\bar{\cal U}} = a^{j}_{\cal U}$, for all $j \in \{1,\ldots,K\}$. Therefore, (\ref{eq:enhance}) can be re-written as
	\begin{equation}
	L_d^*(\mathcal{U}) \geq L_d^*(\bar{\mathcal{U}}) \geq \sum\limits_{j=1}^{K} \frac{a^{j}_{\cal U}}{N}\cdot\frac{K-j}{j+1}.
	\end{equation}
	\end{proof}
	
	Then following the same arguments as in the proof for the minimum uplink communication load, we have
	\begin{equation}
	L_d^*({\cal U})\geq \tfrac{K-\mu K}{\mu K+1}= \tfrac{\mu K}{\mu K+1} \cdot (\tfrac{1}{\mu}-1). 
	\end{equation} 
	
	For general $\frac{1}{K} \leq \mu \leq 1$, $L_d^*({\cal U})$ is lower bounded by the lower convex envelope of the points $\{(\mu,\tfrac{\mu K}{\mu K+1}(\tfrac{1}{\mu}-1)): \mu \in\{\frac{1}{K},\frac{2}{K},...,1\}\}$. 
	
This completes the proof of Theorem~3.
	
\begin{theorem}
	Consider a decentralized wireless distributed computing application. For any random dataset placement with a placement distribution $P$ that achieves an information loss $\Delta$, and communication schemes that achieve communication loads $L_u$ and $L_d$ with high probability for large $N$, $L_u$ and $L_d$ are lower bounded by $\frac{1}{\mu}-1$ when $K$ is large and $\Delta$ approaches $0$.
\end{theorem}

\begin{remark}
When the number of participating users is large (large $K$), the above lower bound in Theorem~4 coincides with the asymptotic loads achieved by the proposed decentralized CWDC scheme stated in Theorem~2 (see Remark~\ref{decastm}). Therefore, the proposed decentralized scheme is asymptotically optimal. $\hfill \square$
\end{remark}

We now prove Theorem 4 by showing that for any decentralized dataset placement, the minimum achievable communication loads are lower bounded by $\tfrac{1}{\mu}-1$ when the number of participating users is large and the information loss approaches zero. Hence, the asymptotic communication loads achieved by the proposed decentralized scheme can not be further improved. In particular, for a particular realization of the dataset placement $\mathcal{U}$ with information loss $\Delta(\mathcal{U})$, we denote the minimum possible uplink and downlink communication loads by $L^*_{\textup{decent},u}(\mathcal{U})$ and $L^*_{\textup{decent},d}(\mathcal{U})$, and derive lower bounds on $L^*_{\textup{decent},u}(\mathcal{U})$ and $L^*_{\textup{decent},d}(\mathcal{U})$ respectively.

We note that given the information loss $\Delta({\cal U})$, $1-\Delta(\mathcal{U})$ fraction of files are available across the participating users, all of which need to be processed to compute the outputs (see (\ref{eq:part})). Among those files, $\bar{\mu}({\cal U}) \triangleq \frac{\mu}{1-\Delta(\mathcal{U})}$ fraction of them are stored by each participating user. Following the same steps in proving the lower bounds of the centralized setting, the minimum communication loads for the dataset placement ${\cal U}$ are lower bounded as follows.
\begin{align}
L^*_{\textup{decent},u}(\mathcal{U}) &\geq \left(\frac{1}{\bar{\mu}({\cal U})}-1\right)\left(1-\Delta(\mathcal{U}) \right)\\
&=\left(\frac{1-\Delta(\mathcal{U})}{\mu}-1\right) \left(1-\Delta(\mathcal{U}) \right), \\ 
L^*_{\textup{decent},d}(\mathcal{U}) &\geq \frac{\bar{\mu}({\cal U}) K}{\bar{\mu}({\cal U})K+1}\left(\frac{1}{\bar{\mu}({\cal U})}-1\right)\left(1-\Delta(\mathcal{U}) \right)\\
&\!=\frac{\mu K}{\mu K+1-\Delta(\mathcal{U})}\!\left(\frac{1-\Delta(\mathcal{U})}{\mu}-\!1\!\right)\!\left(1\!-\!\Delta(\mathcal{U})\right).
\end{align}

Since the above bounds hold for any realization of dataset placement $\mathcal{U}$, for a decentralized dataset placement scheme with a distribution $P$ that achieves an information loss $\Delta(P)$, communication loads $L_{\textup{decent},u}^*(P)$, $L_{\textup{decent},d}^*(P)$ with high probability, the following inequalities hold.
\begin{align}
L_{\textup{decent},u}^*(P) &\geq \left(\frac{1-\Delta(P)}{\mu}-1\right) \left(1-\Delta(P) \right), \\ 
L_{\textup{decent},d}^*(P) &\geq \frac{\mu K}{\mu K+1-\Delta(P)}\!\left(\!\frac{1\!-\!\Delta(P)}{\mu}-\!1\!\right)\!\left(1\!-\!\Delta(P)\right).
\end{align}

Hence, when the number of active users are large, the achievable uplink and downlink communication loads, for any decentralized dataset placement scheme with a distribution $P$ that achieves a vanishing information loss are bounded by
\begin{align}
L_{\textup{decent},u}^*(P) &\geq \frac{1}{\mu}-1, \\ 
L_{\textup{decent},d}^*(P) &\geq \frac{1}{\mu}-1.
\end{align}

This completes the proof of Theorem~4.
 
\section{Conclusions and Future Directions}
In this paper, we proposed a scalable wireless distributed computing framework, for both the centralized and the decentralized settings, such that the shuffling load does not increase with the number of participating users. In particular, we use a repetitive placement of the dataset across the users to enable coding, reducing the shuffling load by a factor that scales linearly with the network size.

In this paper, we have abstracted out several practical challenges to demonstrate the theoretical feasibility of a scalable wireless distributed computing framework and quantify the resulting gains. Future directions can be to generalize the model to also incorporate the following important aspects.

\begin{itemize}
\item \emph{Network Heterogeneity.} Most mobile networks are heterogeneous. Different mobile devices have different link quality, processing power, battery capacity, and QoS requirement. For example, the proposed coded computing schemes are for a set of mobile users with similar uplink/downlink channel strengths and communication rates. When users have heterogeneous link capacities, one straightforward solution is to first partition the users into groups, such that all users within a group have similar channel strength, and then apply the proposed schemes within each group. However, designing the optimal grouping strategy is a challenging problem that requires further exploration. Other than performing the computations at the users themselves, the superior computation and storage capacity of a growing number of edge servers at access points encourage computation offloading to the edge servers. Another interesting problem is to consider a mobile network of heterogeneous users with the possibility of performing computations at the edge servers, and study the optimal scheduling and communication schemes (see e.g., \cite{barbarossa2014communicating}).

\item \emph{Computation Heterogeneity.} In many wireless distributed computing applications (especially for graph processing), the intermediate computation results have heterogeneous sizes. For example, for a navigation application over a highly clustered map, some parts of the map generate much more useful information than the other parts, resulting in highly skewed intermediate results. In such scenario, the proposed coding scheme still applies, but the coding operations are not symmetric as in the case of homogeneous intermediate results (e.g., one may now need to compute the XOR of two data segments with different sizes). Alternatively, we can consider a low-complexity greedy approach, in which we perform the dataset placement to maximize the number of multicasting opportunities that simultaneously deliver useful information to the largest possible number of users. Nevertheless, finding the optimal dataset placement and coding scheme in the case of heterogeneous computation results
remains a challenging open problem.

\item \emph{Straggling/Failing Users.} So far we have assumed that for both the centralized and the decentralized settings, once the collaborative computation process starts, all participating users are active and reliable until the end of the computation. However, similar to the straggler problems in wireline computer clusters (see e.g.,  \cite{zaharia2008improving}), one needs to account for the possibilities of mobile users losing connectivity, leaving, and joining the application in the middle of computation. One approach to deal with straggling/failing users during the computation process is to assign users \emph{coded} computation tasks using e.g., Maximum-Distance-Separable codes (see~\cite{lee2015speeding} for an example of applying coded computations on matrix multiplication). Using this approach, the successful execution of the mobile application can be achieved by retrieving the computation results from only a  subset of ``healthy'' users, and this can provide the system with certain level of robustness to straggling/failing users during the course of computation.

\item \emph{Multi-Stage Computation.} 
We have so far designed the schemes for applications with one stage of Map-Reduce computation. However, a general application contains multiple stages of computations, interconnected as a directed acyclic graph (see e.g.,~\cite{saha2015apache}). It would be interesting to understand the optimal schemes for such general applications. A preliminary exploration along this direction was recently presented in~\cite{LMA-Allerton16}.
\end{itemize} 

\section{Acknowledgement}
This work is in part supported by NSF grants CCF-1408639, NETS-1419632, ONR award N000141612189, NSA grant, and a research gift from Intel. This material is based upon work supported by Defense Advanced Research Projects Agency (DARPA) under Contract No. HR001117C0053. The views, opinions, and/or findings expressed are those of the author(s) and should not be interpreted as representing the official views or policies of the Department of Defense or the U.S. Government.

\bibliographystyle{IEEEtran}
\bibliography{reference-final}

\end{document}